\documentclass[journal]{IEEEtran}
\usepackage{verbatim}
\usepackage{graphicx}
\usepackage{amsmath,amssymb,amsfonts,bm}
\usepackage{subfigure}
\usepackage{tikz}
\usepackage{psfrag}
\usepackage{pstool}
\usepackage{epstopdf}
\usepackage{pgfplots}

\usepackage[dvips]{epsfig}
\usepackage{amsthm}
\usepackage{color}
\usepackage[usenames,dvipsnames]{pstricks}
\usepackage[dvips]{epsfig}
\usepackage{pst-grad} 
\usepackage{pst-plot} 
\usepackage{amsmath,graphicx}
\usepackage{enumerate}
\usepackage{amsbsy}
\usepackage{amssymb}
\usepackage{amsthm}
\usepackage{amscd}
\usepackage{subfigure}
\usepackage{color}
\usepackage[numbers,sort&compress,square]{natbib}

\usepackage{listings}
\usepackage{stackrel}
\usepackage{authblk}
\usepackage[ruled,vlined,linesnumbered,resetcount]{algorithm2e}
\makeatletter
\newcommand{\removelatexerror}{\let\@latex@error\@gobble}
\makeatother
\SetAlFnt{\footnotesize}

\newtheorem{lemm}{Lemma}

\newtheorem{rema}{Remark}

\begin{document}
\title{Energy Efficient Precoder in Multi-User MIMO Systems with Imperfect Channel State Information}
\author{Hossein Vaezy, \and Mohammad Javad Omidi, \and Halim Yanikomeroglu, \textit{Fellow}, \textit{IEEE}\\
h.vaezy@ec.iut.ac.ir, omidi@cc.iut.ac.ir, halim@sce.carleton.ca}
\affil{Department of Electrical and Computer Engineering, Isfahan University of Technology, Isfahan 84156-83111, Iran.}
\affil{Department of Systems and Computer Engineering, Carleton University, Ottawa, ON K1S5B6, Canada.}

\maketitle
\begin{abstract}
This article is on the energy efficient precoder design in multi-user multiple-input multiple-output (MU-MIMO) systems which is also robust with respect to the imperfect channel state information (CSI) at the transmitters. In other words, we design the precoder matrix associated with each transmitter to maximize the general energy efficiency of the network. We investigate the problem in two conventional cases. The first case considers the statistical characterization for the channel estimation error that leads to a quadratically constrained quadratic program (QCQP) with a semi-closed-form solution. Then, we turn our attention to the case which considers only the uncertainty region for the channel estimation error; this case eventually results in a semi-definite program (SDP).
\end{abstract}
\section{Introduction}
Due to the increasing interest in green wireless networks, the energy efficiency (EE) criterion has attracted significant attention in both academia and industry~\cite{5G}. In general, EE is defined as the ratio of the sum-rate\footnote{In this paper, rate means spectral efficiency, in line with the literature on this subject. Please refer to~\eqref{equ4} for the definition of rate.} to the total power consumption and is measured in bits/Hz/Joule\footnote{EE unit can also be expressed by bps/Hz/watt.}. The fractional form of this utility function causes additional difficulty in the design procedure~\cite{5G}.

It is well known that the multiple-input multiple-output (MIMO) technique has a great potential for improving spectral efficiency (SE)~\cite{TseFWC}. EE maximization in a single-input single-output (SISO) setting is investigated in~\cite{Meshpoor2007}. Reference~\cite{chongavtc2011} presents an energy efficient power allocation method in a multiple-input single-output (MISO) broadcast channel (BC) with the random beamforming. Authors in~\cite{cpanxujsac2016} consider MISO interference channel (IC). Single-user MIMO channel is also considered in~\cite{belmegalasaulce2011}. Using non-cooperative games, a distributed method for designing transmit covariance matrix is presented in~\cite{7150420}. 

In practice, perfect channel state information (CSI) may not be available as a result of the quantization error, channel estimation error, feedback delay, etc.~\cite{6497698}. In the literature, imperfect CSI is modeled by either stochastic or deterministic models. The stochastic model assumes that the channel is random with an unknown instantaneous value but the channel statistics, such as the first or second order moments are known at the transmitter. On the other hand, the deterministic model, which is more suitable to characterize instantaneous CSI with errors, assumes that we know some limited information about the error, e.g., the norm of the error. Therefore, in the deterministic model, the actual channel estimation error lies in an uncertainty region~\cite{6497698}.

A minimum mean-squared error (MMSE) precoder which operates based on optimizing the worst-case channel, was proposed in~\cite{1468482}. The worst-case precoding was also studied for MIMO multi-access channels~\cite{4586299}, broadcast channels~\cite{4808242}, multi-cell scenario~\cite{5590310}, and cognitive radio systems~\cite{5164911}.

In this paper, we consider a multi-user MIMO (MU-MIMO) system and design the precoders for two cases: i) the imperfect CSI with the known statistics of the channel error, and ii) the norm-bounded error without any information about the statistical characterization. We develop an energy efficient precoder design for MU-MIMO with robustness against the imperfect CSI.


\section{System model}
\subsection{Channel model}
In this paper we consider a symmetric MIMO-IC consisting of a set of $K$ user-pairs (Fig.~\ref{ICfig}). Users communicate concurrently with their own pairs and cause interference to the other users. We assume that the $k$th transmitter and the $k$th receiver are equipped with $M_{k}$ and $N_{k}$ antennas, respectively. The received signal $\boldsymbol{y}_{n}\in\mathbb{C}^{N_{k}\times 1}$ at the $n$th receiver is given by

\begin{equation} \label{equ1}
\boldsymbol{y}_{n} = \sum_{k=1}^{K} \textbf{H}_{nk}\textbf{V}_{k}\boldsymbol{s}_{k} + \boldsymbol{z}_{n},
\end{equation}
where $\textbf{H}_{jk}\in\mathbb{C}^{N_{j}\times M_{k}}$ denotes the complex channel between the $k$th transmitter and the $j$th receiver, $\textbf{V}_{k}\in\mathbb{C}^{M_{k}\times d_{k}}$ denotes the precoding matrix at the $k$th transmitter, $d_{k}$ represents the number of data symbols for the $k$th user, $\boldsymbol{s}_{k}\in\mathbb{C}^{d_{k}\times 1}$ is the normalized data symbols of the $k$th user (i.e., $\text{E}\big\{\boldsymbol{s}_{k}\boldsymbol{s}_{k}^{H}\big\}=\textbf{I}_{d_{k}}$), and $\boldsymbol{z}_{n}\in\mathbb{C}^{N_{n}\times 1}$ represents the noise vector at the $n$th receiver which is assumed to be a circularly symmetric complex Gaussian random vector with covariance matrix $\sigma^{2}\textbf{I}_{N}$.

Denote $\textbf{U}_{n}\in\mathbb{C}^{N_{n}\times d_{n}}$ as the linear decoder at the $n$th receiver, then the estimated data symbols at the $n$th receiver become
\begin{equation} \label{equ2}
\hat{\boldsymbol{s}}_{n} = \textbf{U}_{n}^{H}\textbf{H}_{nn}\textbf{V}_{n}\boldsymbol{s}_{n} + \sum_{l=1,l\neq n}^{K} \textbf{U}_{n}^{H}\textbf{H}_{nl}\textbf{V}_{l}\boldsymbol{s}_{l} + \textbf{U}_{n}^{H}\boldsymbol{z}_{n}.
\end{equation}
The terms on the right-hand side of~\eqref{equ2} are the desired signal, interference and the filtered noise, respectively. For the sake of simplicity, we assume that the number of antennas at the transmitters, receivers as well as the number of data symbols at various users are the same and are equal to $M$, $N$ and $d$, respectively.

\begin{figure}[h]
\centering
\includegraphics[width=0.8\columnwidth,height=4cm]{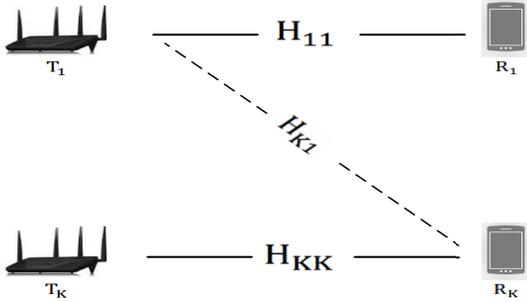}
\caption{A general MIMO-IC with $K$ user-pairs.}
\label{ICfig}
\end{figure}

\subsection{Channel state error model}
In practice, perfect CSI at the transmitters is unavailable due to many factors such as the channel estimation error, quantization error, and feedback error/delay. Here, we model the channel matrix as~\cite{6497698}
\begin{equation}
\textbf{H}_{ij}=\tilde{\textbf{H}}_{ij}+\boldsymbol{\Delta}_{ij},
\end{equation}
where $\boldsymbol{\Delta}_{ij}$ and $\tilde{\textbf{H}}_{ij}$ denote the estimation error and estimated channel between the $j$th transmitter and the $i$th receiver, respectively.
\section{EE maximization in MIMO-IC}
Due to the increasing interest in green wireless networks, our goal is to maximize the general energy efficiency (GEE). GEE is defined as the ratio of the weighted sum-rate to the total power consumption and is given by
\begin{equation}
\textsf{GEE}=\frac{\sum_{k=1}^{K}{\alpha_{k}R_{k}}}{\sum_{k=1}^{K}{(\rho\text{tr}(\textbf{V}_{k}\textbf{V}_{k}^{H})+MP_{cir}})},
\end{equation}
where $\rho$ ($\geq 1$) is inverse of the power amplifier efficiency and $P_{cir}$ denotes the circuit power per antenna for each user due to the digital processor, mixer, etc~\cite{6680785}. Finally, $R_{k}$ is the achievable rate (bps/Hz) of the $k$th user with a weight factor $\alpha_{k}$ and is given by~\cite{7150420}
\begin{align}\label{equ4}
R_{k} = \log\big|\textbf{I}_{d}+\textbf{V}_{k}^{H}\textbf{H}_{kk}^{H}\textbf{U}_{k}
(\textbf{U}_{k}^{H}\textbf{C}_{k}\textbf{U}_{k})^{-1}\textbf{U}_{k}^{H}\textbf{H}_{kk}\textbf{V}_{k}\big|,
\end{align}
where $|.|$ denotes the determinant and $\textbf{C}_{k}$ is the covariance matrix of the interference-plus-noise and is equal to
\begin{equation} \label{equ5}
\textbf{C}_{k} = \sigma^{2}\textbf{I}_{N}+\sum_{l\neq{k}}^{K} \textbf{H}_{kl}\textbf{V}_{l}\textbf{V}_{l}^{H}\textbf{H}_{kl}^{H}.\nonumber
\end{equation}
\subsection{Imperfect CSI with known error statistics}
Here, our aim is to maximize GEE in the presence of channel estimation error with known error statistics. We assume the channel estimation errors are uncorrelated with $\textbf{H}_{ij}$ and their values are i.i.d. zero mean circularly symmetric complex Gaussian random variables with variance $\sigma_{\Delta}^{2}$~\cite{Vaezy}. Therefore, we deal with the following optimization problem:
\begin{align}
&\max_{\{\textbf{V}_{n},\textbf{U}_{n}\}} \quad \textsf{GEE}\label{robust_prob1}\nonumber\\
&\quad \text{s.t.}\quad \text{tr}(\textbf{V}_{l}^{H}\textbf{V}_{l})\leq P_{m}, \quad\forall l=1,\ldots,K,
\end{align}
where $P_{m}$ denotes the available maximum power.

Using Dinkelbach method~\cite{dinkel} we convert the fractional form objective function to a subtractive form as follow:
\begin{align}\label{obj_1}
\sum_{k=1}^{K}{\alpha_{k}R_{k}}-\eta\sum_{k=1}^{K}{\big(\rho\text{tr}(\textbf{V}_{k}\textbf{V}_{k}^{H})+MP_{cir}}\big),
\end{align}
where $\eta$ in each iteration is a constant term and equals to GEE in the previous iteration. The sum-rate in the objective function~\eqref{obj_1} is a non-convex part. Therefore, we use the majorization-minimization (MaMi) technique~\cite{Vaezy} to handle the non-convexity of the problem and finally solve the problem by an iterative method. Using Lemma 1 in~\cite{Vaezy}, properties of Kronecker product, and eliminating constant terms, the objective function becomes
\begin{align}
-\Big\{\sum_{i=1}^{K}{\boldsymbol{x}_{i}^{H}(\textbf{I}_{d}\otimes\boldsymbol{\Psi_{i}}+\eta\rho\textbf{I}_{M\times d})\boldsymbol{x}_{i}+2\Re{(\boldsymbol{x}_{i}^{H}\boldsymbol{b}_{i})}}\Big\}
\end{align}
where,
\begin{align} \label{psi_b_def}
&\boldsymbol{x}_{i}\triangleq\text{vec}(\textbf{V}_{i})\in\mathbb{C}^{Md\times 1}\:, \quad \boldsymbol{b}_{i}\triangleq\text{vec}(\alpha_{i}\tilde{\textbf{H}}_{ii}^{H}\mbox{$\textbf{F}_{i,12}$}^{H})\in\mathbb{C}^{Md\times 1}\nonumber\\
&\boldsymbol{\Psi}_{i}\triangleq\sum_{l=1}^{K} \alpha_{l}\Big\{\tilde{\textbf{H}}_{li}^{H}\textbf{F}_{l,22}\tilde{\textbf{H}}_{li}+\sigma_{\Delta}^{2}
\text{tr}(\textbf{F}_{l,22})\textbf{I}_{M}\Big\}\in\mathbb{C}^{M\times M}
\end{align}
and
\begin{align}
&\textbf{F}_{i,22}=\textbf{F}_{i,12}^{H}(\textbf{I}_{d}+\tilde{\textbf{V}}_{i}^{H}\tilde{\textbf{H}}_{ii}^{H}\textbf{C}_{i}^{-1}\tilde{\textbf{H}}_{ii}\tilde{\textbf{V}}_{i})^{-1}\textbf{F}_{i,12},\nonumber\\
&\textbf{F}_{i,12}=-\tilde{\textbf{V}}_{i}^{H}\tilde{\textbf{H}}_{ii}^{H}\textbf{C}_{i}^{-1}\nonumber.
\end{align}
In the above, $\tilde{\textbf{V}_{i}}$ denotes the value of $\textbf{V}_{i}$ in the previous iteration, $\Re(.)$ represents the real part of a complex value and $\text{vec}(.)$ stands for the vectorization operation.
Finally, the optimization problem is decomposed to $K$ smaller optimization problem, e.g., the $i$th subproblem becomes
\begin{align}\label{QCQP_robust}
&\min_{\boldsymbol{x}_{i}\in\mathbb{C}^{Md\times 1}} \mbox{$\boldsymbol{x}_{i}$}^{H} (\textbf{I}_{d}\otimes\boldsymbol{\Psi}_{i}+\eta\rho\textbf{I}_{Md})\boldsymbol{x}_{i}
+2\Re (\mbox{$\boldsymbol{x}_{i}$}^{H}\boldsymbol{b}_{i}) \nonumber\\
& \quad {\text{s. t.}}\quad \|{\boldsymbol{x}}_{i}\|{_{2}^{2}}\leq P_{i}.
\end{align}
The optimization problem in~\eqref{QCQP_robust} is a quadratically constrained quadratic program (QCQP) which can be solved using the Lagrange multipliers method in closed-form. Therefore, the optimal $\boldsymbol{x}_{i}$ is equal to~\cite{Vaezy}
\begin{equation} \label{equ23}
\boldsymbol{x}_{i}(\lambda)=-(\textbf{I}_{d}\otimes\boldsymbol{\Psi}_{i}+\eta\rho\textbf{I}_{Md}+\lambda\textbf{I}_{Md})^{-1}\boldsymbol{b}_{i},\nonumber
\end{equation}
where, the Lagrange multiplier, $\lambda$ can be computed by solving the following non-linear equation:
\begin{equation} \label{equ25}
\mbox{$\boldsymbol{b}_{i}$}^{H}(\textbf{I}_{d}\otimes\boldsymbol{\Psi}_{i}+\eta\rho\textbf{I}_{Md}+\lambda\textbf{I}_{Md})^{-2}\boldsymbol{b}_{i}=P_{i}.\nonumber
\end{equation}
\begin{rema}
The computational complexity of QCQP problem in~\eqref{QCQP_robust} is $\mathcal{O}(\max(M,N)^{3})$~\cite{Vaezy}.
\end{rema}
\subsection{Imperfect CSI with norm-bounded error}
Here, our target is to maximize GEE in a MIMO-IC while considering the robustness in the presence of the channel uncertainty which is considered as an elliptical region. Hence, the problem can be formulated as
\begin{align}
&\max_{\{\textbf{V}_{n},\textbf{U}_{n}\}} \min_{\{\boldsymbol{\Delta}_{ij}\}}\quad \textsf{GEE}\label{robust_prob1}\nonumber\\
&\quad \text{s.t.}\quad \text{tr}(\textbf{V}_{l}^{H}\textbf{V}_{l})\leq P_{m}, \quad\forall l=1,\ldots,K\nonumber \\
&\quad\quad \quad ||\textbf{B}_{ij}\boldsymbol{\Delta}_{ij}||_{2}^{2}\leq\epsilon_{ij}, \quad \forall i,j=1,\ldots , K,
\end{align}
where, $\textbf{B}_{ij}\succeq 0$ and $\epsilon_{ij}\geq 0$ jointly represent the feasible ellipsoid region for $\boldsymbol{\Delta}_{ij}$.

From Dinkelbach method~\cite{dinkel}, we can assume that there exists an equivalent objective function in a subtractive form. As a result, with a fixed $\eta$, the fractional optimization problem~\eqref{robust_prob1} can be solved by focusing on the non-convex problem below in a subtractive form:
\begin{align}\label{eq_robust2}
&\max_{\{\textbf{V}_{n},\textbf{U}_{n}\}} \min_{\{\boldsymbol{\Delta}_{ij}\}}\quad \sum_{i=1}^{K}{\alpha_{i}R_{i}}-\eta \sum_{j=1}^{K}{\big(\rho\text{tr}(\textbf{V}_{j}^{H}\textbf{V}_{j})+MP_{cir}\big)}\nonumber\\
&\quad \text{s.t.}\quad \text{tr}(\textbf{V}_{l}^{H}\textbf{V}_{l})\leq P_{m}, \quad\forall l=1,\ldots,K\nonumber \\
&\quad\quad \quad ||\textbf{B}_{ij}\boldsymbol{\Delta}_{ij}||_{2}^{2}\leq\epsilon_{ij}, \quad \forall i,j=1,\ldots , K.
\end{align}
\begin{lemm}
Let $\textbf{S}\in\mathbb{C}^{d\times d}$ and $\textbf{W}\in \mathbb{C}^{d\times d}$ denote the positive semi-definite matrices. Then, $\log_{2} | \textbf{S}|$ is an optimal value of the following problem:
\begin{align}\label{lemm11}
\max_{\textbf{W}} \:\: \log_{2} | \textbf{W}\ln2|-\text{tr}(\textbf{S}\textbf{W})+\frac{d}{\ln2}.
\end{align}
\end{lemm}
\begin{proof}
By checking the first order optimality condition for the convex problem~\eqref{lemm11} we can recognize (interested readers may refer to~\cite{ChristensenAgarwal2008}).
\end{proof}
Using Lemma 1 and the relation between mean-squared error (MSE) and rate~\cite{ChristensenAgarwal2008}, the optimization problem~\eqref{eq_robust2} can be rewritten as
\begin{align}\label{eq_robust3}
&\max_{\{\textbf{V}_{n},\textbf{U}_{n}\}} \min_{\{\boldsymbol{\Delta}_{ij}\}} \max_{\{\textbf{W}_{n}\}}\; \sum_{i=1}^{K}{\Big \{\alpha_{i}\log_{2} | \textbf{W}_{i}\ln2|-\alpha_{i}\text{tr}(\textbf{W}_{i}\textbf{MSE}_{i})}\nonumber\\
&\quad\quad\quad\quad\quad\quad\quad+\frac{\alpha_{i}d}{\ln 2}\Big\}-\eta \sum_{j=1}^{K}{\big(\rho\text{tr}(\textbf{V}_{j}^{H}\textbf{V}_{j})+MP_{cir}\big)}\nonumber\\
&\quad \text{s.t.}\quad \text{tr}(\textbf{V}_{l}^{H}\textbf{V}_{l})\leq P_{m}, \quad\forall l=1,\ldots,K\nonumber \\
&\quad\quad \quad ||\textbf{B}_{ij}\boldsymbol{\Delta}_{ij}||_{2}^{2}\leq\epsilon_{ij}, \forall i,j=1,\ldots , K,
\end{align}
where,
\begin{align}
\textbf{MSE}_{i} &= (\textbf{U}_{i}^{H}\textbf{H}_{ii}\textbf{V}_{i}-\textbf{I}_{d})(\textbf{V}_{i}^{H}\textbf{H}_{ii}^{H}\textbf{U}_{i}-\textbf{I}_{d})^{H}\nonumber\\
&+\sum_{j\neq i}^{K}{\textbf{U}_{i}^{H}\textbf{H}_{ij}\textbf{V}_{j}\textbf{V}_{j}^{H}\textbf{H}_{ij}^{H}\textbf{U}_{i}}+\sigma^{2}\textbf{U}_{i}^{H}\textbf{U}_{i}.
\end{align}

Due to three maximization minimization, problem~\eqref{eq_robust3} is still hard to solve. Therefore, we replace the inner min-max by max-min version which can be known as the lower bound of the objective function.

Also, in order to simplify the objective function, we define $\textbf{W}_{i}=\textbf{G}_{i}\textbf{G}_{i}^{H}$ and rewrite $\text{tr}(\textbf{W}_{i}\textbf{MSE}_{i})$ as a sum of $K$ vector norms. Finally, by substituting $\textbf{H}_{ij}=\tilde{\textbf{H}}_{ij}+\boldsymbol{\Delta}_{ij}$ and using some algebraic operations we have
\begin{equation}
\text{tr}\{\textbf{W}_{i}\textbf{MSE}_{i}\}=\sum_{j=1}^{K}{\big|\big| \textbf{e}_{ij}+\textbf{E}_{ij}\text{vec}(\boldsymbol{\Delta}_{ij}) \big|\big|_{2}^{2}},
\end{equation}
where, $\textbf{e}_{ij}\in \mathbb{C}^{l_{ij}}$ and $\textbf{E}_{ij}\mathbb{C}^{l_{ij}\times MN}$ are defined as
\begin{align}
\textbf{e}_{ii}=
\begin{bmatrix}
\text{vec}\big((\textbf{V}_{i}^{H}\tilde{\textbf{H}}_{ii}^{H}\textbf{U}_{i}-\textbf{I}_{d})\textbf{G}_{i}\big) \\
\sigma\text{vec}(\textbf{U}_{i}\textbf{G}_{i})
\end{bmatrix}
,\: \textbf{E}_{ii}=
\begin{bmatrix}
\textbf{G}_{i}^{H}\textbf{U}_{i}^{H}\otimes\textbf{V}_{i}^{H} \\
\textbf{0}_{Nd\times NM}
\end{bmatrix},\nonumber\\
\textbf{e}_{ij}=\text{vec}\big(\textbf{V}_{j}^{H}\tilde{\textbf{H}}_{ij}^{H}\textbf{U}_{i}\textbf{G}_{i}\big),\: \textbf{E}_{ij}=
\textbf{G}_{i}^{H}\textbf{U}_{i}^{H}\otimes\textbf{V}_{j}^{H},\quad j\neq i ,\nonumber
\end{align}
and $l_{ij}=d^{2}$ for the different users ($i\neq j$) and $l_{ij}=d^{2}+Nd$ for the same users case ($i=j$).
Due to the inner minimization problem in~\eqref{eq_robust3}, the later optimization problem is still hard to solve. Therefore, by introducing the slack variables $\lambda_{ij}$, we deal with the following optimization problem:
\begin{align}\label{eq_robust5}
&\max_{\{\overline{\textbf{V}},\overline{\textbf{U}},\overline{\textbf{G}},\overline{\lambda}\}} \; \sum_{i=1}^{K}{\Big\{2\alpha_{i}\log_{2} | \textbf{G}_{i}|-\alpha_{i}\sum_{j=1}^{K}{\lambda_{ij}}}-\eta\rho\text{tr}(\textbf{V}_{i}^{H}\textbf{V}_{i})\Big\}\nonumber\\
&\quad \text{s.t.}\quad \text{tr}(\textbf{V}_{i}^{H}\textbf{V}_{i})\leq P_{m}, \quad\forall i=1,\ldots,K \\
&\quad\quad \quad \big|\big|\textbf{B}_{ij}\boldsymbol{\Delta}_{ij}\big|\big|_{2}^{2}\leq\epsilon_{ij}, \quad \forall i,j=1,\ldots , K \nonumber\\
&\quad\quad \quad \big|\big| \textbf{e}_{ij}+\textbf{E}_{ij}\text{vec}(\boldsymbol{\Delta}_{ij}) \big|\big|_{2}^{2}\leq \lambda_{ij},\: \forall i,j=1,\ldots, K.\nonumber
\end{align}

To show two norm-bounded constraints in~\eqref{eq_robust5} in the form of one constraint and further simplify the optimization problem, we use the following lemma:
\begin{lemm}
~\cite{Eldar_robust} Given matrices $\textbf{P}, \textbf{Q}, \textbf{A}$ with $\textbf{A}=\textbf{A}^{H}$
\begin{equation}
\textbf{A} \succeq \textbf{P}^{H}\textbf{X}\textbf{Q}+\textbf{Q}^{H}\textbf{X}^{H}\textbf{P}, \quad \forall \textbf{X}:||\textbf{X}|| \leq \epsilon
\end{equation}
if and only if there exists a $\mu \geq 0$ such that
\begin{align}
\begin{bmatrix}
\textbf{A}-\mu \textbf{Q}^{H}\textbf{Q} & -\epsilon \textbf{P}^{H} \\
-\epsilon \textbf{P} & \mu\textbf{I}
\end{bmatrix}
\succeq 0.
\end{align}
\end{lemm}
\begin{proof}
See~\cite{Eldar_robust}.
\end{proof}
Now, we choose the matrices in Lemma 2 as
\begin{align}
\textbf{A}=
\begin{bmatrix}
\lambda_{ij} & \textbf{e}_{ij}^{H} \\
\textbf{e}_{ij} & \textbf{I}_{l_{ij}}
\end{bmatrix}
,\textbf{P}=
\begin{bmatrix}
\textbf{0}_{MN\times 1} & \tilde{\textbf{B}}_{ij}^{H}\textbf{E}_{ij}^{H}
\end{bmatrix}\nonumber
,\textbf{Q}=
\begin{bmatrix}
-1 & \textbf{0}_{1\times l_{ij}}
\end{bmatrix}
,
\end{align}
and
\begin{align}
\textbf{X}=\text{vec}(\tilde{\boldsymbol{\Delta}}_{ij}^{H}),\;\;\tilde{\boldsymbol{\Delta}}_{ij}\triangleq\textbf{B}_{ij}\boldsymbol{\Delta}_{ij}\:,\;\; \tilde{\textbf{B}}_{ij}\triangleq\textbf{B}_{ij}^{-1}\otimes \textbf{I}_{M}.\nonumber
\end{align}
Finally, the optimization problem~\eqref{eq_robust5} can be rewritten using Schur complement as
\begin{align}\label{eq_robust6}
&\max_{\{\overline{\textbf{V}},\overline{\textbf{U}},\overline{\textbf{G}},\overline{\lambda},\overline{\mu}\}} \; \sum_{i=1}^{K}{\Big\{2\alpha_{i}\log_{2} | \textbf{G}_{i}|-\alpha_{i}\sum_{j=1}^{K}{\lambda_{ij}}}-\eta \rho\text{tr}(\textbf{V}_{i}^{H}\textbf{V}_{i})\Big\}\nonumber\\
&\text{s.t.}\quad\begin{bmatrix}
P_{m} & \text{vec}(\textbf{V}_{i})^{H}\\
\text{vec}(\textbf{V}_{i}) & \textbf{I}_{Md}
\end{bmatrix}
\succeq 0,\: \forall i= 1,\ldots,K\\
&\begin{bmatrix}
\lambda_{ij}-\mu_{ij} & \textbf{e}_{ij}^{H} & \textbf{0}_{1\times MN}\\
\textbf{e}_{ij} & \textbf{I}_{l_{ij}} & -\epsilon_{ij}\textbf{E}_{ij}\tilde{\textbf{B}}_{ij}\\
\textbf{0}_{MN \times 1} & -\epsilon_{ij}\tilde{\textbf{B}}_{ij}^{H}\textbf{E}_{ij}^{H} & \mu_{ij}\textbf{I}_{MN}
\end{bmatrix}
\succeq 0,\forall i,j =1,\ldots,K. \nonumber
\end{align}
The problem presented in~\eqref{eq_robust6} is still non-convex jointly over $\overline{\textbf{V}},\overline{\textbf{U}}$, and $\overline{\textbf{G}}$. However, we can observe that it is convex over the individual variables. Therefore, we use an iterative method in which we fix two variables and optimize for the remaining one in each iteration. For given $\overline{\textbf{U}}$ and $\overline{\textbf{G}}$, we end up having a quadratic objective function with linear constraints with respect to $\overline{\textbf{V}}$ which can be solved efficiently by CVX tool. By fixing $\overline{\textbf{V}}$ and $\overline{\textbf{G}}$, the problem results in a semi-definite program (SDP). Finally, using the MAX-DET algorithm and $log\underline{}det$ command in CVX, we can compute $\overline{\textbf{G}}$~\cite{Boyd:2004:CO:993483}.
\begin{rema}
The computational complexity of the problem in~\eqref{eq_robust6} dominated by that of the SDP problem which is known to be $\mathcal{O}(\max(M,N)^{6.5})$~\cite{1275676}.
\end{rema}
\section{Numerical results} \label{sec:simulation}
In this section, we numerically evaluate the GEE of the proposed methods in the presence of channel estimation error. In all simulations, the circularly symmetric complex Gaussian noise vector is assumed with zero mean and covariance matrix $\textbf{I}_{N}$. The MIMO-IC is modeled by a Rayleigh flat fading and a spatially uncorrelated channel with variance $\sigma_{h}^{2}=1$,~\cite{TseFWC}. In the simulations, the power amplifier efficiency is $0.38$. 
\begin{figure}[h]
\centering
\includegraphics[trim = 10mm 1mm -5mm 2mm,clip,width=1.1\columnwidth,height=6cm]{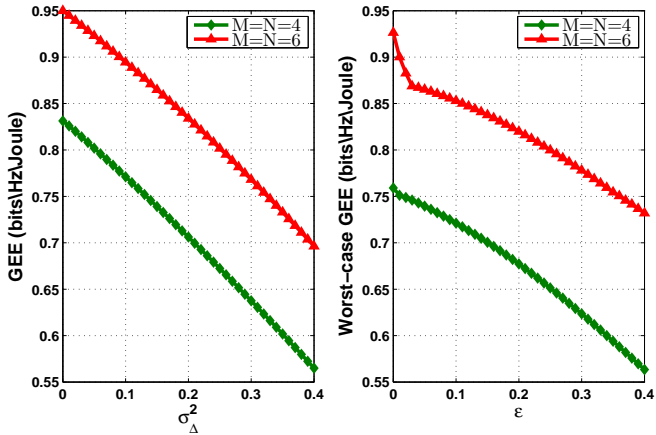}
\caption{GEE for imperfect channel with statistical model (left side) and uncertainty region (right side) with $P_{cir}=-5$ dBW and $P_m=0$ dBW .}
\label{sub_uncer_statis}
\end{figure}

Fig.~\ref{sub_uncer_statis} shows GEE for varying error region $\epsilon$ and error variance $\sigma_{\Delta}^{2}$. For this study, we set $K=3$, $d=1$ and consider two different values of $M$, i.e., $M=4,6$. It is clear that the GEE decreases with an increase in $\sigma_{\Delta}^{2}$ and $\epsilon$. It should be noted that the GEE in the robust method with the uncertainty region is in fact the worst-case (lower bound) for the GEE. It is revealed that by increasing the number of antennas the circuit power increases in a linear manner. On the positive side, this increase in the number of antennas results in more degrees of freedom in the precoder design. Here, the contribution of the precoder matrix is more than circuit power. It is worthy to note that the corner point in the right sub-figure in Fig.~\ref{sub_uncer_statis} is due to the degradation of the second constraint in~\eqref{eq_robust6} to $
\begin{bmatrix}
\lambda_{ij}-\mu_{ij} & \textbf{e}_{ij}^{H} \\
\textbf{e}_{ij} & \textbf{I}_{l_{ij}}
\end{bmatrix}
\succeq 0$, when $\epsilon$ is small.
\section{Summary}
In this paper we investigated the general energy efficiency in the multi-user MIMO networks. The effect of channel estimation error has been studied in the presence of known and unknown error statistics. The optimization tools are used to convert the non-convex fractional problems into the equivalent subtractive form and convex problems. The effectiveness of the proposed methods is illustrated via simulation in different error statistics and uncertainty regions.

\bibliographystyle{IEEEtran}
\bibliography{IEEEabrv,refer}

\begin{thebibliography}{10}
\providecommand{\url}[1]{#1}
\csname url@samestyle\endcsname
\providecommand{\newblock}{\relax}
\providecommand{\bibinfo}[2]{#2}
\providecommand{\BIBentrySTDinterwordspacing}{\spaceskip=0pt\relax}
\providecommand{\BIBentryALTinterwordstretchfactor}{4}
\providecommand{\BIBentryALTinterwordspacing}{\spaceskip=\fontdimen2\font plus
\BIBentryALTinterwordstretchfactor\fontdimen3\font minus
  \fontdimen4\font\relax}
\providecommand{\BIBforeignlanguage}[2]{{%
\expandafter\ifx\csname l@#1\endcsname\relax
\typeout{** WARNING: IEEEtran.bst: No hyphenation pattern has been}%
\typeout{** loaded for the language `#1'. Using the pattern for}%
\typeout{** the default language instead.}%
\else
\language=\csname l@#1\endcsname
\fi
#2}}
\providecommand{\BIBdecl}{\relax}
\BIBdecl

\bibitem{5G}
J.~G. Andrews, S.~Buzzi, W.~Choi, S.~V. Hanly, A.~Lozano, A.~C.~K. Soong, and
  J.~C. Zhang, ``What will 5{G} be?'' \emph{{IEEE} J. Sel. Areas Commun.},
  vol.~32, no.~6, pp. 1065--1082, Jun. 2014.

\bibitem{TseFWC}
D.~Tse and P.~Viswanath, \emph{Fundamentals of Wireless Communication}.\hskip
  1em plus 0.5em minus 0.4em\relax New York, NY, USA: Cambridge University
  Press, 2005.

\bibitem{Meshpoor2007}
F.~Meshkati, H.~V. Poor, and S.~C. Schwartz, ``Energy-efficient resource
  allocation in wireless networks,'' \emph{{IEEE} Signal Process. Mag.},
  vol.~24, no.~3, pp. 58--68, May. 2007.

\bibitem{chongavtc2011}
Z.~Chong and E.~A. Jorswieck, ``Energy efficiency in random opportunistic
  beamforming,'' in \emph{IEEE 73rd Vehicular Technology Conference
  (VTC2011-Spring)}, May. 2011, pp. 1--5.

\bibitem{cpanxujsac2016}
C.~Pan, W.~Xu, J.~Wang, H.~Ren, W.~Zhang, N.~Huang, and M.~Chen,
  ``Pricing-based distributed energy-efficient beamforming for {MISO}
  interference channels,'' \emph{{IEEE} J. Sel. Areas Commun.}, vol.~34, no.~4,
  pp. 710--722, Apr. 2016.

\bibitem{belmegalasaulce2011}
E.~V. Belmega and S.~Lasaulce, ``Energy-efficient precoding for
  multiple-antenna terminals,'' \emph{{IEEE} Trans. Signal Process.}, vol.~59,
  no.~1, pp. 329--340, Jan. 2011.

\bibitem{7150420}
C.~Pan, W.~Xu, J.~Wang, H.~Ren, W.~Zhang, N.~Huang, and M.~Chen, ``Totally
  distributed energy-efficient transmission in {MIMO} interference channels,''
  \emph{{IEEE} Trans. Wireless Commun.}, vol.~14, no.~11, pp. 6325--6338, Nov.
  2015.

\bibitem{6497698}
J.~Wang, M.~Bengtsson, B.~Ottersten, and D.~P. Palomar, ``Robust {MIMO}
  precoding for several classes of channel uncertainty,'' \emph{{IEEE} Trans.
  Signal Process.}, vol.~61, no.~12, pp. 3056--3070, Jun. 2013.

\bibitem{1468482}
Y.~Guo and B.~C. Levy, ``Worst-case {MSE} precoder design for imperfectly known
  {MIMO} communications channels,'' \emph{{IEEE} Trans. Signal Process.},
  vol.~53, no.~8, pp. 2918--2930, Aug. 2005.

\bibitem{4586299}
M.~B. Shenouda and T.~N. Davidson, ``On the design of linear transceivers for
  multi-user systems with channel uncertainty,'' \emph{{IEEE} J. Sel. Areas
  Commun.}, vol.~26, no.~6, pp. 1015--1024, Aug. 2008.

\bibitem{4808242}
N.~Vucic, H.~Boche, and S.~Shi, ``Robust transceiver optimization in downlink
  multiuser {MIMO} systems,'' \emph{{IEEE} Trans. Signal Process.}, vol.~57,
  no.~9, pp. 3576--3587, Sep. 2009.

\bibitem{5590310}
A.~Tajer, N.~Prasad, and X.~Wang, ``Robust linear precoder design for
  multi-cell downlink transmission,'' \emph{{IEEE} Trans. Signal Process.},
  vol.~59, no.~1, pp. 235--251, Jan. 2011.

\bibitem{5164911}
G.~Zheng, K.~K. Wong, and B.~Ottersten, ``Robust cognitive beamforming with
  bounded channel uncertainties,'' \emph{{IEEE} Trans. Signal Process.},
  vol.~57, no.~12, pp. 4871--4881, Dec. 2009.

\bibitem{6680785}
S.~He, Y.~Huang, L.~Yang, and B.~Ottersten, ``Coordinated multicell multiuser
  precoding for maximizing weighted sum energy efficiency,'' \emph{{IEEE}
  Trans. Signal Process.}, vol.~62, no.~3, pp. 741--751, Feb. 2014.

\bibitem{Vaezy}
H.~Vaezy, M.~M. Naghsh, M.~J. Omidi, and E.~H.~M. Alian, ``Efficient transmit
  covariance design in {MIMO} interference channel,'' \emph{{IEEE} Trans. Veh.
  Technol.}, vol.~67, no.~7, pp. 5793--5805, Jul. 2018.

\bibitem{dinkel}
W.~Dinkelbach, ``On nonlinear fractional programming,'' \emph{Management
  Science}, vol.~13, no.~7, pp. 492--498, 1967.

\bibitem{ChristensenAgarwal2008}
S.~S. Christensen, R.~Agarwal, E.~D. Carvalho, and J.~M. Cioffi, ``{Weighted
  sum-rate maximization using weighted MMSE for MIMO-BC beamforming design},''
  \emph{IEEE Trans. Wireless Commun.}, vol.~7, no.~12, pp. 4792--4799, Dec.
  2008.

\bibitem{Eldar_robust}
Y.~C. Eldar, A.~Ben-Tal, and A.~Nemirovski, ``Robust mean-squared error
  estimation in the presence of model uncertainties,'' \emph{{IEEE} Trans.
  Signal Process.}, vol.~53, no.~1, pp. 168--181, Jan. 2005.

\bibitem{Boyd:2004:CO:993483}
S.~Boyd and L.~Vandenberghe, \emph{{Convex Optimization}}.\hskip 1em plus 0.5em
  minus 0.4em\relax New York, NY, USA: Cambridge University Press, 2004.

\bibitem{1275676}
Z.-Q. Luo, T.~N. Davidson, G.~B. Giannakis, and K.~M. Wong, ``Transceiver
  optimization for block-based multiple access through {ISI} channels,''
  \emph{{IEEE} Trans. Signal Process.}, vol.~52, no.~4, pp. 1037--1052, Apr.
  2004.

\end{thebibliography}
\end{document}